\documentclass[sigconf]{acmart}

\AtBeginDocument{
  }

\usepackage{amsmath}

\usepackage{amssymb}
\usepackage{amsthm}
\usepackage[ruled,vlined]{algorithm2e}
\usepackage{subcaption}
\usepackage{multirow}
\newtheorem{theorem}{Theorem}
\newtheorem{lemma}[theorem]{Lemma}
\newtheorem{proposition}[theorem]{Proposition}
\usepackage{ulem}
\newcommand{\sellname}{GSPRec}
\usepackage{pifont}   
\newcommand{\cmark}{\ding{51}} 
\newcommand{\xmark}{\ding{55}}
\usepackage{tabularx}
\usepackage{subcaption}

\usepackage{rotating}

\copyrightyear{2026}
\acmYear{2026}
\setcopyright{cc}
\setcctype{by}
\acmConference[RecSys '26]{20th ACM Conference on Recommender Systems}{September 27-October 02, 2026}{Minneapolis, MN, USA}
\acmBooktitle{20th ACM Conference on Recommender Systems (RecSys '26), September 27-October 02, 2026, Minneapolis, MN, USA}
\acmDOI{10.1145/3773078.3831790}
\acmISBN{979-8-4007-2284-4/2026/09}

\begin{document}

\title{GSPRec: On Improving Item Representations in Graph Signal Processing for Collaborative Filtering}

\author{Ahmad Bin Rabiah}
\orcid{0009-0002-5991-4299}
\affiliation{
  \institution{University of California, San Diego}
  \city{La Jolla}
  \state{CA}
  \country{USA}
}
\email{abinrabiah@ucsd.edu}

\author{Julian McAuley}
\orcid{0000-0003-0955-7588}
\affiliation{
  \institution{University of California, San Diego}
  \city{La Jolla}
  \state{CA}
  \country{USA}
}
\email{jmcauley@eng.ucsd.edu}

\begin{abstract} 
Graph-based collaborative filtering methods act as low-pass filters in the spectral domain and discard the intermediate-frequency components where community-level user preferences reside. 
Existing GSP-based methods address the loss through sophisticated filter designs, yet derive item representations from the user-item interaction matrix alone. 
The interaction matrix captures which items each user interacted with, but not which items appear close together in users' interaction sequences.
We propose \sellname{}, a graph spectral collaborative filtering framework that produces richer item spectral representations by incorporating item-item proximity derived from user interaction ordering before spectral filtering. 
\sellname{} derives item-item edges from user interaction ordering and strengthens the edges through multi-hop diffusion with exponential decay. The unified graph topology incorporates the diffused edges alongside user-item interactions.
The resulting Laplacian exposes intermediate-frequency structure that a Gaussian bandpass filter selectively amplifies.
A low-pass filter retains broad popularity trends. 
Experiments on four real-world datasets show that \sellname{} outperforms all graph CF baselines, with average improvements of 5.12\% in NDCG@10.
Ablation studies establish that graph construction and filter design are coupled.
\sellname{} without the bandpass filter falls below every GSP baseline, whereas GSPRec without item-item proximity still surpasses baseline.
\end{abstract}



\maketitle

\section{Introduction}
Graph-based collaborative filtering (CF) methods exploit high-order connectivity in the user-item interaction graph~\cite{wang2019neural,he2020lightgcn,liu2021interest}.
These methods propagate signals through multi-hop paths to reveal indirect relationships between users and items that direct interactions alone cannot capture. 
Nevertheless, existing methods, whether based on message passing or spectral filtering, derive item representations from the user-item interaction matrix alone~\cite{shen2021powerful,guo2023manipulating, liu2023personalized,xia2024hierarchical, xia2024frequency,qin2024polycf}. 
The interaction matrix captures which items each user interacted with, but not which items users interacted with together.

Graph signal processing (GSP) offers a spectral view of collaborative filtering~\cite{ortega2018graph, kruzick2017graph,shen2021powerful}. 
Prior work~\cite{shen2021powerful, dong2019learning} shows that neighborhood aggregation in CF methods acts as a low-pass filter.
Low-pass filtering amplifies items with broad, population-wide popularity and suppresses intermediate-frequency components where community-level user preferences reside~\cite{shen2021powerful}. Low-pass behavior has motivated subsequent work on filter design~\cite{shen2021powerful,guo2023manipulating, liu2023personalized,xia2024hierarchical,xia2024frequency}. Yet existing filters operate on item representations derived from the same source without making them richer. 
Enriching item representations with external signals consistently improves recommendation~\cite{he2016vbpr,wang2019kgat,wei2024llmrec,rabiah25bridging}. 
However, external signals require auxiliary data sources, whereas interaction ordering already exists within the interaction data.
The interaction matrix records that a user interacted with two items, but not whether those interactions occurred in close succession or far apart in the user's history. Incorporating ordering information into the graph topology produces item spectral representations that capture not just who interacted with each item, but which items were interacted with together across user histories. 
The enriched representations expose intermediate-frequency spectral structure absent from the interaction matrix alone.
\begin{figure*}[t] 
    \centering 
    \begin{subfigure}{0.31\textwidth} 
        \includegraphics[width=\linewidth]{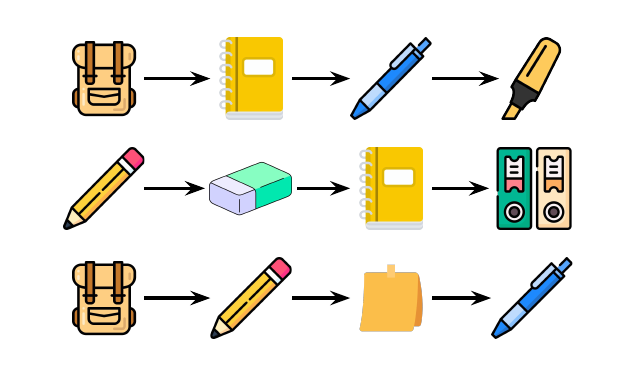}
        \Description{An ordered sequence of items from one user's interaction history, drawn as icons connected by arrows that indicate chronological order.}
        \caption{Sequential pattern}
        \label{fig:seq}
    \end{subfigure} \hfill 
    \begin{subfigure}{0.31\textwidth} 
        \includegraphics[width=\linewidth]{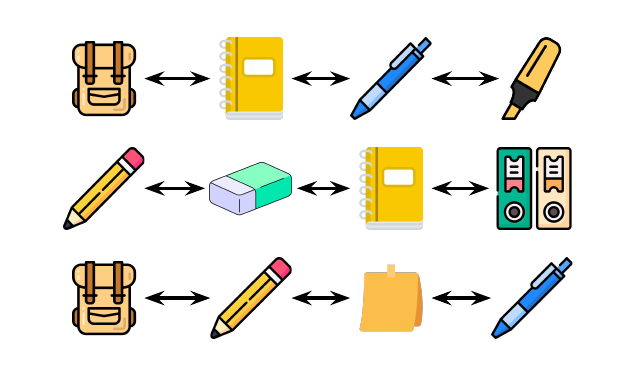}
        \Description{The same item sequence with double-headed arrows between adjacent items. Symmetrization removes interaction order and keeps bidirectional proximity between neighboring items.}
        \caption{Bidirectional pattern} 
        \label{fig:bidir}
    \end{subfigure} \hfill 
    \begin{subfigure}{0.31\textwidth} 
        \includegraphics[width=\linewidth]{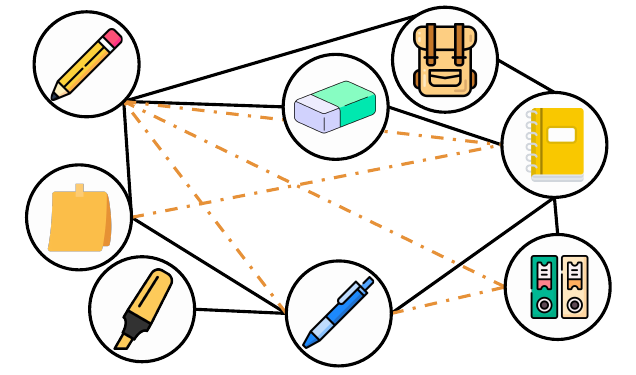}
        \Description{The unified item graph after multi-hop diffusion. Solid lines mark direct item-item edges and dashed lines mark weaker higher-order edges added by diffusion.}
        \caption{Item Graph with Multi-hop Diffusion}
        \label{fig:undir}
    \end{subfigure}
    \caption{Graph construction from user histories. (a) Ordered user interaction history. (b) Item proximity graph $\mathbf{S}'$, obtained by symmetrizing the sequence in (a), in which edges connect items that appear together in user histories regardless of interaction order. Items that appear close together in a history receive stronger connections than items that merely share users. (c) Unified item graph after multi-hop diffusion, where solid lines show direct edges and dashed lines show higher-order structure. Items that appear close together in a history receive stronger connections than items that merely share users.}
    \label{fig:seq2graph} 
\end{figure*}

We propose \textbf{\sellname{}}, a graph spectral framework for collaborative filtering that produces richer item spectral representations by incorporating item-item proximity derived from user interaction ordering before spectral filtering.
\sellname{} derives item-item edges from user interaction ordering and strengthens the edges through multi-hop diffusion with exponential decay~\cite{gasteiger2019diffusion}.
The unified graph topology incorporates the diffused edges, as Figure~\ref{fig:seq2graph} illustrates.
The resulting graph Laplacian yields item spectral representations that capture proximity structure from user histories alongside user-item interaction signals.
A Gaussian bandpass filter~\cite{hammond2011wavelets, tremblay2018design} selectively amplifies the intermediate-frequency structure the enriched representations expose. A low-pass filter retains broad popularity trends.
Graph construction and filter design in GSP-based CF are coupled. The intermediate-frequency structure that richer representations introduce is precisely what low-pass filtering suppresses.

Experiments on four real-world datasets demonstrate that \sellname{} consistently outperforms all graph CF baselines, with average improvements of 5.12\% in NDCG@10.
We summarize our contributions as follows:
\begin{itemize} 
    \item We propose a graph construction method that produces richer item spectral representations by incorporating item-item proximity derived from user interaction ordering, introducing intermediate-frequency spectral structure absent from existing GSP-based CF graphs. To our knowledge, \sellname{} is the first method to incorporate ordering-derived proximity into item spectral representations in GSP-based CF.
    \item We design a Gaussian bandpass filter that selectively amplifies the intermediate-frequency structure richer item representations expose. It outperforms all existing GSP methods even without the new graph construction.
    \item We show that graph construction and filter design are coupled in GSP-based CF. Enriched item representations degrade performance in the absence of the bandpass filter, while joint design of both components achieves state-of-the-art results across all datasets and metrics.
    \item We provide theoretical analysis establishing convergence of the diffusion process and validity of the unified Laplacian.
\end{itemize}
\section{Background \& Related Work}

\subsection{Graph Signal Processing Background} 
In a user-item graph, popular items produce globally smooth signals that correspond to low-frequency spectral components~\cite{shen2021powerful}. 
Community-specific preferences vary between user groups while remaining consistent within them~\cite{von2007tutorial,newman2006finding}. 
Such patterns occupy the intermediate-frequency band~\cite{shen2021powerful,xia2024hierarchical}.
High-frequency components reflect noise and sparse outlier behavior~\cite{ortega2018graph,shuman2013emerging}. 
The graph Laplacian makes the spectral properties precise. 
Formally, let $\mathbf{x} \in \mathbb{R}^{|\mathcal{V}|}$ denote a graph signal defined on $\mathcal{G} = (\mathcal{V}, \mathcal{E})$. 
The normalized Laplacian $\mathbf{L} = \mathbf{I} - \mathbf{D}^{-1/2}\mathbf{A}\mathbf{D}^{-1/2}$ is a real symmetric positive semidefinite matrix~\cite{chung1997spectral}. 
It admits the eigendecomposition $\mathbf{L} = \mathbf{U}\boldsymbol{\Lambda}\mathbf{U}^T$, where $\mathbf{U}$ is the orthonormal eigenvector matrix and $\boldsymbol{\Lambda} = \mathrm{diag}(\lambda_1, \ldots, \lambda_n)$ with $0 \leq \lambda_1 \leq \ldots \leq \lambda_n \leq 2$. 
Spectral filtering transforms $\mathbf{x}$ via a frequency response function $g$~\cite{ortega2018graph}:
\begin{equation}
    \mathbf{x}' = \mathbf{U} \, g(\boldsymbol{\Lambda}) \, 
    \mathbf{U}^T \mathbf{x}, \quad \text{where } 
    g(\boldsymbol{\Lambda}) = \mathrm{diag}(g(\lambda_1), 
    \ldots, g(\lambda_n))
\end{equation}
The choice of $g$ determines which preference signals are preserved or discarded in the recommendation process.

\subsection{Graph-Based Collaborative Filtering} 
Message passing methods such as NGCF~\cite{wang2019neural}, LightGCN~\cite{he2020lightgcn}, LR-GCCF~\cite{chen2020revisiting}, and IMP-GCN~\cite{liu2021interest} learn user and item representations by aggregating neighborhood information over the user-item bipartite graph. 
GF-CF~\cite{shen2021powerful} shows that neighborhood aggregation in message passing methods is equivalent to low-pass filtering in the spectral domain. Low-pass filtering therefore suppresses intermediate-frequency components that encode community-level user preferences.
Recent work addresses over-smoothing~\cite{chen2020measuring}, popularity bias~\cite{zhang2023mitigating}, and efficiency~\cite{liang2024lightweight} but operates within the same message passing paradigm. 
\sellname{} departs from the message passing paradigm by operating directly on the graph Laplacian spectrum, incorporating ordering-derived item-item proximity into the graph before filtering.
\begin{figure*}[t]
    \centering
    \includegraphics[width=\textwidth]{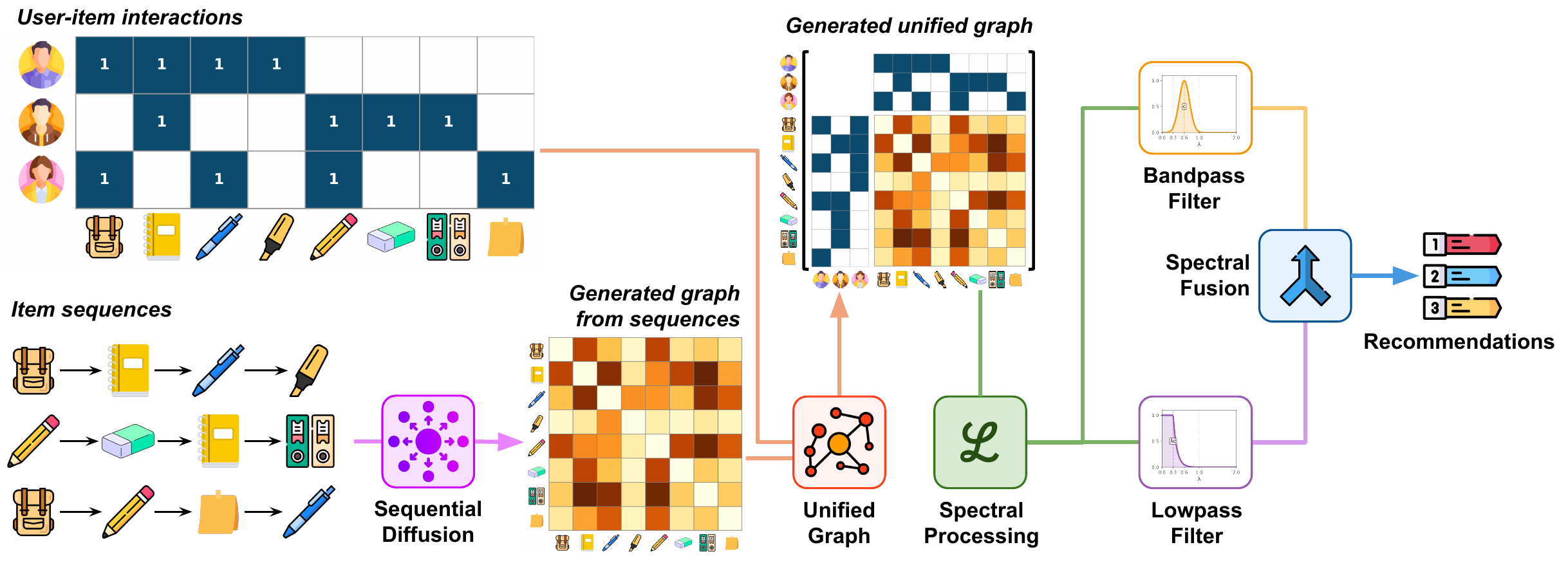}
    \Description{Pipeline diagram of GSPRec. User-item interactions and item sequences enter graph construction. Sequential diffusion builds an item graph from the sequences, which merges with the interactions into a unified graph and its Laplacian. A bandpass filter and a lowpass filter process the spectrum, and spectral fusion combines both outputs into recommendations.}
    \caption{Overview of \sellname{}. Starting from user-item interactions and user histories, \sellname{} constructs a unified graph by incorporating ordering-derived item-item proximity into the graph topology via multi-hop diffusion with exponential decay~(Eq.~\ref{eq:diffusion}). A Gaussian bandpass filter targets the intermediate-frequency structure the unified graph introduces. A complementary low-pass filter retains broad popularity trends. The final recommendation score fuses both signals.}
    \label{fig:overview}
\end{figure*}
\subsection{GSP-based Collaborative Filtering}
\begin{table}[t]
    \centering
    \scriptsize
    \setlength{\tabcolsep}{4pt}
    \caption{\sellname{} is the only general CF method that incorporates ordering-derived item-item proximity with intermediate-frequency spectral filtering.}
    \label{tab:method_comparison}
    \begin{tabular}{lccccc}
        \toprule
        \textbf{Method} & \textbf{Task} & \textbf{Spectral} & \textbf{History} & \textbf{Graph} & \textbf{Mid-} \\
        & & \textbf{Filter} & \textbf{Prox.} & \textbf{Structure} & \textbf{Freq} \\
        \midrule
        \multicolumn{6}{l}{\textit{Message Passing GCNs:}} \\
        NGCF~\cite{wang2019neural}        & General & \xmark & \xmark & User-Item & \xmark \\
        LightGCN~\cite{he2020lightgcn}    & General & \xmark & \xmark & User-Item & \xmark \\
        LR-GCCF~\cite{chen2020revisiting} & General & \xmark & \xmark & User-Item & \xmark \\
        IMP-GCN~\cite{liu2021interest}    & General & \xmark & \xmark & User-Item & \xmark \\
        \midrule
        \multicolumn{6}{l}{\textit{Spectral Filtering Methods (GSP-based):}} \\
        GF-CF~\cite{shen2021powerful}     & General & LP& \xmark & User-Item & \xmark \\
        JGCF~\cite{guo2023manipulating}   & General & HP& \xmark & User-Item & \xmark \\
        PGSP~\cite{liu2023personalized}   & General & LP & \xmark & User-Item & \xmark \\
        HiGSP~\cite{xia2024hierarchical}  & General & LP & \xmark & User-Item & \xmark \\
        FaGSP~\cite{xia2024frequency}     & General & LP & \xmark & User-Item & \xmark \\
        PolyCF~\cite{qin2024polycf} & General & LP &  \xmark & User-Item & \xmark \\
        \midrule
        \textbf{\sellname{} (Ours)} & \textbf{General} & \textbf{BP+LP} & \cmark & \textbf{UI+Hist} & \cmark \\
        \bottomrule
    \end{tabular}
    \vspace{-10pt}
\end{table}
GSP-based recommenders treat user-item interactions as graph signals and apply spectral filters to extract preference information~\cite{huang2017collaborative,shen2021powerful}. 
GF-CF~\cite{shen2021powerful} establishes a unified spectral framework for collaborative filtering and shows that existing CF methods correspond to different low-pass filters applied to the user-item interaction graph. 
Subsequent methods propose a range of filter designs: PGSP~\cite{liu2023personalized} introduces personalized graph signals with mixed-frequency filtering; HiGSP~\cite{xia2024hierarchical} designs hierarchical filters at cluster and global levels; FaGSP~\cite{xia2024frequency} combines cascaded high-pass and parallel low-pass modules. 
JGCF~\cite{guo2023manipulating} proposes Jacobi polynomial-based filters for joint low- and high-frequency modeling. 
PolyCF~\cite{qin2024polycf} learns optimal polynomial filter shapes through joint graph optimization and pairwise ranking objectives. 
All existing GSP-based recommenders derive item representations from the user-item interaction matrix alone, as summarized in Table~\ref{tab:method_comparison}. 
\sellname{} augments the interaction matrix with ordering-derived item-item proximity via multi-hop diffusion. The enriched graph exposes intermediate-frequency spectral structure that existing representations do not contain, and a Gaussian bandpass filter~\cite{hammond2011wavelets,tremblay2018design} selectively amplifies the exposed structure. No prior GSP-based CF method enriches item spectral representations with interaction ordering.
\section{\sellname{}}
\sellname{} addresses collaborative filtering through two components: (1) a graph construction procedure that incorporates ordering-derived item-item proximity into the item spectral representations, and (2) a spectral filtering framework that applies complementary bandpass and low-pass filters to the resulting Laplacian.

\subsection{Problem Formulation}
Let $\mathcal{U} = \{u_1, \ldots, u_m\}$ and  $\mathcal{I} = \{i_1, \ldots, i_n\}$ denote the user  and item sets. 
Historical interactions are recorded as  timestamped triplets $\mathcal{D} = \{(u, i, t) \mid  u \in \mathcal{U}, i \in \mathcal{I}, t \in  \mathbb{R}^+\}$. 
The binary user-item interaction  matrix $\mathbf{X} \in \{0,1\}^{m \times n}$ is  derived from $\mathcal{D}$. 
An entry $x_{ui} = 1$ if  user $u$ has interacted with item $i$, and $x_{ui} = 0$  otherwise. 
For each user $u$, timestamps in $\mathcal{D}$  induce a chronologically ordered item sequence  $\mathcal{S}_u = [i_1, i_2, \ldots]$ with  $t_1 < t_2 < \cdots$. 
Our task is collaborative  filtering: given $\mathbf{X}$ and $\{\mathcal{S}_u\}$,  predict unobserved user-item interactions.

\subsection{Recommendation Graph Construction}
\label{sec:graph_construction}
We construct a unified bipartite graph that encodes both user-item interactions and item-item proximity derived from user interaction ordering. 
We derive $\mathbf{S}'$ from user interaction ordering, apply multi-hop diffusion to obtain $\mathbf{S}^{(d)}$, normalize to $\tilde{\mathbf{S}}$, and integrate with $\mathbf{X}$ to form the unified Laplacian $\mathbf{L}$.

\subsubsection{Initial Item-Item Graph ($\mathbf{S}'$):} 
We construct a global directed transition matrix $\mathbf{S} \in \{0,1\}^{n \times n}$, where $s_{ij} = 1$ if item $i$ directly precedes item $j$ in any user's interaction sequence. 
For collaborative filtering, the relevant signal is not transition order but item-item proximity, that is, which items appear close together in user interaction ordering.
We therefore form the undirected version $\mathbf{S}' \in \{0,1\}^{n \times n}$, where an edge $(i,j)$ exists if $s_{ij} = 1$ or $s_{ji} = 1$. 
The proximity signal is exactly the component missing from $\mathbf{X}^T\mathbf{X}$. While $\mathbf{X}$ captures that a user interacted with items $i$ and $j$ regardless of timing, $\mathbf{S}'$ captures which items appear close together in user interaction ordering~\cite{sarwar2001item}. 
In addition, spectral filtering requires a symmetric Laplacian to guarantee real eigenvalues and orthogonal eigenvectors~\cite{chung1997spectral, ortega2018graph}.
For example, consider a user sequence \texttt{backpack}$\rightarrow$\texttt{notebook} $\rightarrow$\texttt{pen}. 
Symmetrization creates undirected edges (\texttt{backpack}, \texttt{notebook}) and (\texttt{notebook}, \texttt{pen}). 
The second-order diffusion step (Eq.~\ref{eq:diffusion}) then adds a weighted edge (\texttt{backpack}, \texttt{pen}), which encodes the two-hop item-item proximity between \texttt{backpack} and \texttt{pen} in the user interaction ordering.

\subsubsection{Multi-Hop Diffusion ($\mathbf{S}^{(d)}$) and Normalization ($\tilde{\mathbf{S}}$):} 
\begin{figure}
    \centering
    \includegraphics[width=0.8\linewidth,keepaspectratio]{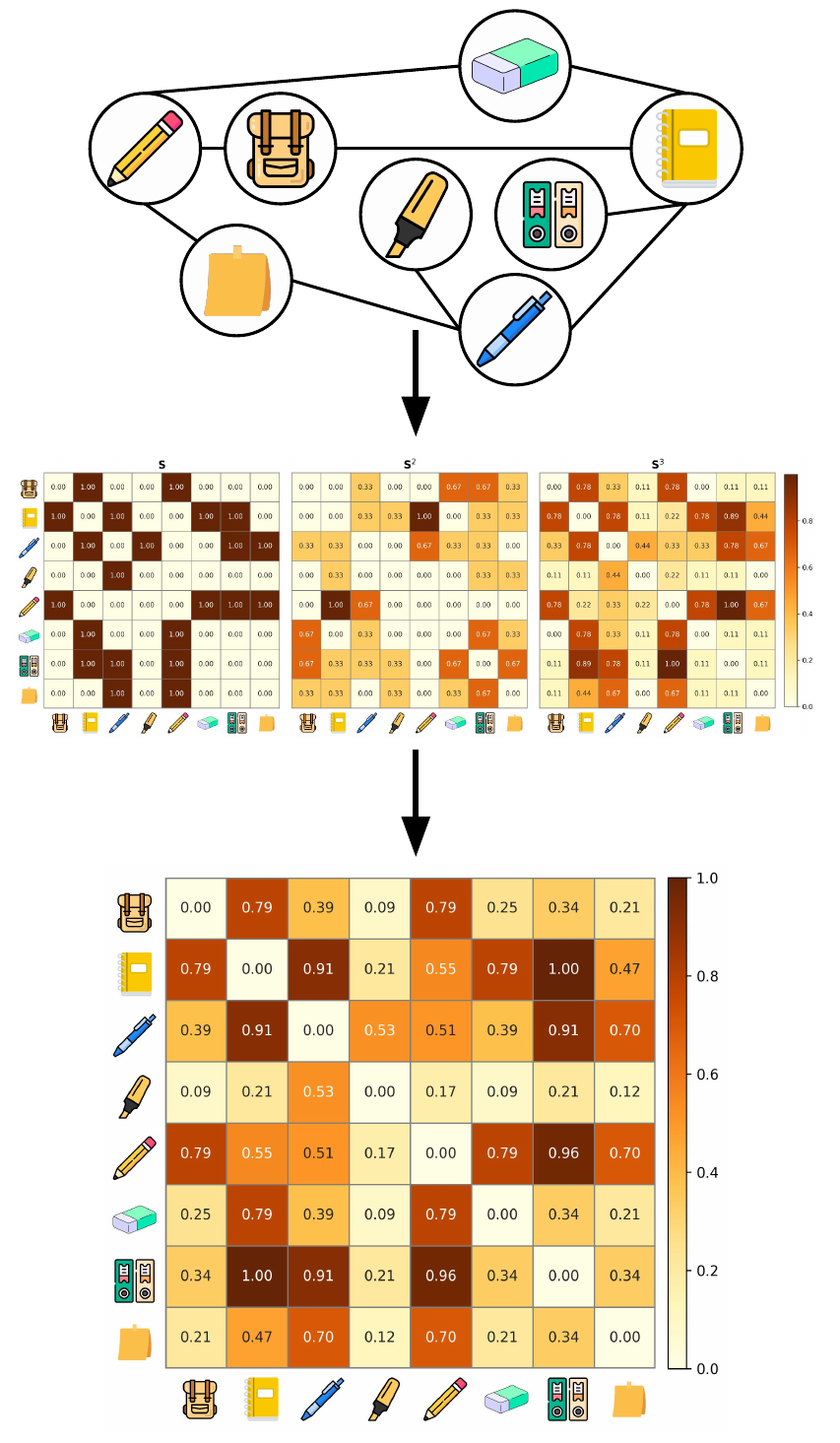}
    \Description{Three matrices illustrating diffusion. Top, the binary adjacency matrix of direct item-item edges. Middle, matrix powers that add higher-order proximity. Bottom, the diffused matrix with exponential decay, in which darker cells mark item pairs that appear closer together in user histories.}
    \caption{Multi-hop diffusion encodes item-item proximity derived from user interaction ordering as edge weight. Top: direct item-item edges $\mathbf{S}'$. Middle: higher-order proximity via $(\mathbf{S}')^k$. Bottom: $\mathbf{S}^{(3)}$ with exponential decay, where darker cells indicate stronger proximity. The construction discards interaction order and preserves item-item proximity as weight magnitude.}
    \label{fig:diffusion-wrap}
    \vspace{-10pt}
\end{figure}
Direct edges in $\mathbf{S}'$ capture only immediate item-item adjacencies.
Items two or more hops apart in $\mathbf{S}'$ also carry meaningful proximity signal.
In the running example, \texttt{backpack} and \texttt{pen} are not directly connected in $\mathbf{S}'$ but are two hops apart via \texttt{notebook}. 
To encode multi-hop proximity, we apply diffusion:
\begin{equation} 
    \mathbf{S}^{(d)} = \sum_{k=1}^{d} \alpha^{k-1} (\mathbf{S}')^k \label{eq:diffusion} 
\end{equation}
Here $\alpha \in (0,1)$ assigns weight $\alpha^{k-1}$ to paths of length $k$, so longer paths contribute exponentially less. 
$d$ controls the maximum diffusion depth. 
Since $\mathbf{S}'$ is symmetric, each power $(\mathbf{S}')^k$ is symmetric, and $\mathbf{S}^{(d)}$ is therefore symmetric. 
Each entry $S^{(d)}_{ij}$ accumulates weighted path counts between items $i$ and $j$ at all lengths up to $d$. 
Items appearing closer together in user interaction ordering receive larger entries.
Figure~\ref{fig:diffusion-wrap} illustrates how proximity weight decays with interaction distance.
The diffusion process is analogous to graph diffusion kernels for similarity definition~\cite{gasteiger2019diffusion}.
\begin{lemma}[Diffusion Convergence] 
    \label{lemma:diffusion_convergence} 
    The diffusion process in Eq.~\ref{eq:diffusion} converges as $d \to \infty$ if $\alpha < 1/\rho(\mathbf{S}')$, where $\rho(\cdot)$ denotes the spectral radius. 
\end{lemma} 
\begin{proof} 
    The series $\mathbf{S}^{(\infty)} = \mathbf{S}'\sum_{j=0}^{\infty}(\alpha\mathbf{S}')^j$ converges when $\rho(\alpha\mathbf{S}') = \alpha\rho(\mathbf{S}') < 1$, requiring $\alpha < 1/\rho(\mathbf{S}')$. 
\end{proof}
In practice, we estimate $\rho(\mathbf{S}')$ via power iteration and set $\alpha$ to satisfy $\alpha < 1/\rho(\mathbf{S}')$. 
We symmetrically normalize $\mathbf{S}^{(d)}$ to scale its weights for stable integration into the unified graph Laplacian~\cite{chung1997spectral}:
\begin{equation} 
    \tilde{\mathbf{S}} = \mathbf{D}_S^{-1/2} \mathbf{S}^{(d)} \mathbf{D}_S^{-1/2}, \quad \text{where } \mathbf{D}_S = \text{diag}(\mathbf{S}^{(d)}\mathbf{1}_n) \label{eq:diffusion-norm} 
\end{equation}

\subsubsection{Interaction Normalization:}
We normalize the user-item interaction matrix $\mathbf{X}$ to account for varying user activity levels and item popularities: 
\begin{equation} 
    \tilde{\mathbf{X}}_U = \mathbf{D}_U^{-1/2}\mathbf{X}, \quad \tilde{\mathbf{X}}_I = \mathbf{X}\mathbf{D}_I^{-1/2} \label{eq:interaction-norm} 
\end{equation} 
where $\mathbf{D}_U = \text{diag}(\mathbf{X}\mathbf{1}_n)$ and $\mathbf{D}_I = \text{diag}(\mathbf{X}^T\mathbf{1}_m)$. 
We use the binary matrix $\mathbf{X}$ directly in $\mathbf{A}$ because the Laplacian normalization in Eq.~\ref{eq:laplacian} accounts for degree differences during spectral filtering. The spectral filters in Section~\ref{sec:filtering} use $\tilde{\mathbf{X}}_U$ and $\tilde{\mathbf{X}}_I$.

\subsubsection{Unified Graph Adjacency and Laplacian ($\mathbf{A}$, $\mathbf{L}$):} 
We integrate user-item interactions and the ordering-derived item-item graph into a single symmetric adjacency matrix:
\begin{equation} 
    \mathbf{A} = 
    \begin{bmatrix} 
        \mathbf{0}_{m \times m} & \mathbf{X} \\ \mathbf{X}^T & \tilde{\mathbf{S}} 
    \end{bmatrix} 
    \label{eq:adjacency} 
\end{equation} 
The $\mathbf{0}_{m \times m}$ block reflects the absence of direct user-user edges. 
Explicit user-user edges provide no additional information, since user-user similarity is fully determined by $\mathbf{X}$ and is incorporated in the low-pass branch through the similarity term $\mathbf{C}_U$, whereas $\tilde{\mathbf{S}}$ encodes a signal that $\mathbf{X}$ does not contain.
The $\mathbf{X}$ block encodes user-item interactions. 
$\tilde{\mathbf{S}}$ encodes item-item proximity derived from user interaction ordering.
From $\mathbf{A}$, we compute the symmetric normalized Laplacian: 

\begin{equation} 
    \mathbf{L} = \mathbf{I} - \mathbf{D}^{-1/2} \mathbf{A} \mathbf{D}^{-1/2}, \quad \text{where } \mathbf{D} = \text{diag}(\mathbf{A}\mathbf{1}) 
\label{eq:laplacian} 
\end{equation}
\begin{proposition}[Laplacian Validity] 
    \label{prop:laplacian_validity} 
    The Laplacian $\mathbf{L}$ constructed from $\mathbf{A}$ is symmetric, positive semidefinite, and has real eigenvalues in $[0, 2]$. 
\end{proposition} 
\begin{proof}[Proof sketch] 
    $\tilde{\mathbf{S}}$ is symmetric by construction. $\mathbf{X}^T$ is the transpose of $\mathbf{X}$, so $\mathbf{A}$ is symmetric. All entries of $\mathbf{X}$ and $\tilde{\mathbf{S}}$ are non-negative, so $\mathbf{A}$ is non-negative. The standard properties of the normalized Laplacian~\cite{chung1997spectral} applied to this symmetric, non-negative $\mathbf{A}$ establish symmetry, positive semidefiniteness, and eigenvalues in $[0,2]$. 
\end{proof}
The unified Laplacian introduces intermediate-frequency spectral structure absent from the standard user-item graph.
The bandpass filter in Section~\ref{sec:filtering} targets the introduced structure.
Proposition~\ref{prop:laplacian_validity} guarantees bounded eigenvalues in $[0,2]$, ensuring the bandpass filter targets a well-defined spectral region across all datasets. 
Ordering-derived item-item edges connect items that frequently appear close together in user histories. 
The added edges create community-level structure in the item-item block of $\mathbf{A}$. 
The community structure reshapes the spectrum of $\mathbf{L}$ and exposes intermediate-frequency components absent from the standard user-item graph.
\begin{figure*}[t]
    \centering
    \includegraphics[width=.9\linewidth]{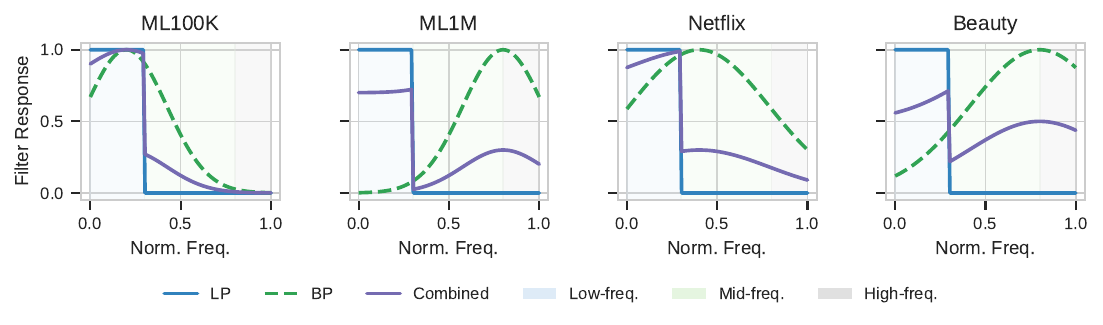}
    \Description{Line plots of filter response versus normalized frequency for ML100K, ML1M, Netflix, and Beauty. The lowpass curve peaks at zero frequency and decays, the bandpass curve peaks in the intermediate band, and the combined curve covers both regions. Shaded areas mark the low, mid, and high frequency bands.}
    \vspace{-10pt}
    \caption{Frequency response of the low-pass (LP), bandpass (BP), and fused filters across all four datasets. The bandpass filter selectively amplifies intermediate-frequency components associated with community-level user preferences. The fused filter captures complementary signals across the spectrum.}
    \label{fig:filter_response} 
\end{figure*}

\subsection{Spectral Filtering Framework}
\label{sec:filtering}
The unified Laplacian admits the eigendecomposition $\mathbf{L} = \mathbf{U}\boldsymbol{\Lambda} \mathbf{U}^T$, where $\boldsymbol{\Lambda} = \text{diag}(\lambda_1, \ldots, \lambda_{m+n})$ and $\mathbf{U} = [\mathbf{u}_1, \ldots, \mathbf{u}_{m+n}]$. 
Incorporating $\tilde{\mathbf{S}}$ into $\mathbf{A}$ reshapes the spectral domain of $\mathbf{L}$ relative to the standard user-item graph. 
In spectral graphs, intermediate-frequency components encode community structure, specifically signals that vary between groups but remain consistent within them~\cite{von2007tutorial,newman2006finding}. 
In the unified spectrum of $\mathbf{L}$, the intermediate-frequency components correspond to community-level user preferences that a Gaussian bandpass filter can target directly. 
A complementary low-pass filter retains broad popularity trends. 
Their outputs are fused to produce the final recommendation scores.

\subsubsection{Bandpass Filtering:} 
We compute a truncated eigendecomposition retaining the $r$ smallest eigenpairs of $\mathbf{L}$. 
We discard high-frequency components, which encode noise~\cite{shen2021powerful,ortega2018graph}. Let $\mathbf{U}_I \in \mathbb{R}^{n \times r}$ be the submatrix of $\mathbf{U}$ corresponding to the $n$ item nodes of $\mathbf{L}$. 
To motivate the filter design, consider decomposing a user preference signal over the eigenbasis of $\mathbf{L}$: $\mathbf{x}_u = \sum_i \hat{x}_i \mathbf{u}_i$, where $\hat{x}_i = \mathbf{u}_i^T\mathbf{x}_u$ are the graph Fourier coefficients. 
Figure~\ref{fig:filter_response} shows the distribution of $|\hat{x}_i|$ across eigenvalue bands for all four datasets: intermediate-frequency components consistently carry higher energy than low or high-frequency components. 
The energy concentration motivates a bandpass filter targeting the intermediate band. We tune the center $c$ and width $w$ per dataset via grid search on the validation set.
We design a Gaussian bandpass filter to target the intermediate-frequency band:
\begin{equation} 
    g_{\text{BP}}(\lambda) = \exp\left(-\frac{(\bar{\lambda} - c)^2}{w}\right), \quad \bar{\lambda} = \frac{\lambda - \lambda_{\min}}{\lambda_{\max} - \lambda_{\min}} \label{eq:bp-kernel} 
\end{equation}
where $c$ and $w$ control the center and width of the frequency band. 
$\lambda_{\min}$ and $\lambda_{\max}$ are the minimum and maximum eigenvalues of the truncated set, so $\bar{\lambda} \in [0,1]$. 
The filter achieves maximum response at $\bar{\lambda} = c$ and attenuates components with eigenvalues far from $c$~\cite{hammond2011wavelets,tremblay2018design}. 
Let $\mathbf{G}_{\text{BP}} = \text{diag}(g_{\text{BP}} (\lambda_1), \ldots, g_{\text{BP}}(\lambda_r))$. The bandpass filtered user-item representation is:
\begin{equation} 
    \mathbf{F}_{\text{BP}} = \mathbf{X}\mathbf{D}_I^{-1/2} \mathbf{U}_I\mathbf{G}_{\text{BP}}\mathbf{U}_I^T \mathbf{D}_I^{-1/2} 
    \label{eq:bandpass} 
\end{equation}

\subsubsection{Lowpass Filtering:}
The low-pass component captures broad popularity trends by projecting augmented user signals onto the low-frequency eigenspace of $\mathbf{L}$. 
Let $\mathbf{C}_U = \tilde{\mathbf{X}}_U \tilde{\mathbf{X}}_U^T \in \mathbb{R}^{m \times m}$ denote the user-user similarity matrix. 
We construct the augmented signal matrix $\mathbf{X}_b = [\mathbf{C}_U, \mathbf{X}] \in \mathbb{R}^{m \times (m+n)}$ by horizontally concatenating $\mathbf{C}_U$ with $\mathbf{X}$. $\mathbf{C}_U$ does not appear in $\mathbf{L}$. 
We introduce it to restore user-user collaborative signals that the zero block in $\mathbf{A}$ omits.
The low-pass filtered representation is: 
\begin{equation} 
    \tilde{\mathbf{F}}_{\text{LP}} = \mathbf{X}_b \mathbf{D}_b^{-1/2}\mathbf{U}\mathbf{U}^T \mathbf{D}_b^{1/2} 
    \label{eq:lowpass} 
\end{equation} 
Here $\mathbf{D}_b = \text{diag}(\mathbf{X}_b^T \mathbf{1}_m)$. The term $\mathbf{U}\mathbf{U}^T$ projects the augmented signal onto the $r$-dimensional low-frequency eigenspace of $\mathbf{L}$. 
The $\mathbf{D}_b^{1/2}$ term restores the original column scale after filtering, following the convention in GF-CF~\cite{shen2021powerful}. 
Since $\mathbf{X}_b$ contains both user-user and user-item columns, we extract only the item-related columns: 
\begin{equation} 
    \mathbf{F}_{\text{LP}} = \tilde{\mathbf{F}}_{\text{LP}} [:, m:] \in \mathbb{R}^{m \times n} 
    \label{eq:lowpass-extract} 
\end{equation}

\subsubsection{Score Fusion:} 
The final recommendation score fuses bandpass and lowpass signals: 
\begin{equation} 
    \mathbf{Y} = \phi \cdot \mathbf{F}_{\text{BP}} +  (1 - \phi) \cdot \mathbf{F}_{\text{LP}}, \quad  \phi \in [0, 1] 
    \label{eq:fusion} 
\end{equation} 
where $\phi$ is a hyperparameter balancing  personalized and global signals, tuned on the  validation set. Both $\mathbf{F}_{\text{BP}}$ and  $\mathbf{F}_{\text{LP}}$ are linear in the input  signals. 
Their convex combination remains a valid  linear operator, as the space of spectral filters  is closed under addition and scalar  multiplication~\cite{sandryhaila2013discrete, puschel2006algebraic}.
\begin{algorithm}[t]
    \caption{\sellname{}}
    \label{alg:gsprec}
    \KwIn{$\mathcal{D}$, $d$, $\alpha$, $r$, $c$, $w$, $\phi$}
    \KwOut{$\mathbf{Y} \in \mathbb{R}^{m \times n}$}
    \textbf{// Graph Construction} \\
    $\mathbf{X} \gets$ binary user-item matrix from $\mathcal{D}$ \\
    $\mathcal{S}_u \gets$ sorted sequences \\
    $\mathbf{S} \gets$ directed transition matrix from $\{\mathcal{S}_u\}$ \\
    $\mathbf{S}' \gets \text{Symmetrize}(\mathbf{S})$ \\
    $\mathbf{S}^{(d)} \gets \sum_{k=1}^{d}\alpha^{k-1}(\mathbf{S}')^k$ \tcp*[r]{Eq.~\ref{eq:diffusion}}
    $\tilde{\mathbf{S}} \gets \mathbf{D}_S^{-1/2}\mathbf{S}^{(d)}\mathbf{D}_S^{-1/2}$ \tcp*[r]{Eq.~\ref{eq:diffusion-norm}}
    $\mathbf{A} \gets \begin{bmatrix} \mathbf{0} & \mathbf{X} \\ 
    \mathbf{X}^T & \tilde{\mathbf{S}} \end{bmatrix}$ \tcp*[r]{Eq.~\ref{eq:adjacency}}
    $\mathbf{L} \gets \mathbf{I} - \mathbf{D}^{-1/2}\mathbf{A}\mathbf{D}^{-1/2}$ \tcp*[r]{Eq.~\ref{eq:laplacian}}
    $[\mathbf{U}, \boldsymbol{\Lambda}] \gets r$ smallest eigenpairs of $\mathbf{L}$ \\
    $\mathbf{U}_I \gets \mathbf{U}[m{:}, :]$ \\
    \textbf{// Bandpass Filtering} \\
    $\bar{\boldsymbol{\lambda}} \gets (\operatorname{diag}(\boldsymbol{\Lambda}) - \lambda_{\min}\mathbf{1}_r)/(\lambda_{\max} - \lambda_{\min})$ \tcp*[r]{Eq.~\ref{eq:bp-kernel}}
    $\mathbf{G}_{\text{BP}} \gets \text{diag}\!\left(\exp\!\left(
    -(\bar{\boldsymbol{\lambda}}-c)^2/w\right)\right)$ \\
    $\mathbf{F}_{\text{BP}} \gets \mathbf{X}\mathbf{D}_I^{-1/2}
    \mathbf{U}_I\mathbf{G}_{\text{BP}}\mathbf{U}_I^T\mathbf{D}_I^{-1/2}$ \tcp*[r]{Eq.~\ref{eq:bandpass}}
    \textbf{// Low-Pass Filtering} \\
    $\mathbf{C}_U \gets \tilde{\mathbf{X}}_U\tilde{\mathbf{X}}_U^T$;\quad
    $\mathbf{X}_b \gets [\mathbf{C}_U, \mathbf{X}]$ \\
    $\tilde{\mathbf{F}}_{\text{LP}} \gets \mathbf{X}_b\mathbf{D}_b^{-1/2}
    \mathbf{U}\mathbf{U}^T\mathbf{D}_b^{1/2}$ \tcp*[r]{Eq.~\ref{eq:lowpass}}
    $\mathbf{F}_{\text{LP}} \gets \tilde{\mathbf{F}}_{\text{LP}}[:, m:]$ \tcp*[r]{Eq.~\ref{eq:lowpass-extract}}
    \textbf{// Fusion} \\
    $\mathbf{Y} \gets \phi \cdot \mathbf{F}_{\text{BP}} + 
    (1-\phi) \cdot \mathbf{F}_{\text{LP}}$ \tcp*[r]{Eq.~\ref{eq:fusion}}
    \Return{$\mathbf{Y}$}
\end{algorithm}
\section{Experiments}
\label{sec:baselines}
We evaluate \sellname{} to answer three questions: 
\textbf{RQ1:} Does \sellname{} outperform 
state-of-the-art CF methods? \textbf{RQ2:} Does 
intermediate-frequency bandpass filtering 
outperform low-pass and high-pass alternatives? 
\textbf{RQ3:} Does incorporating item-item proximity from user interaction ordering provide additional gains beyond bandpass filtering alone?

\subsection{Experimental Setup}
\label{sec:setup}
\subsubsection{Datasets}
We evaluate on four real-world 
datasets spanning interaction densities from 
0.07\% to 6.30\% and scales from 100K to 5.7M 
interactions: ML100K, 
ML1M~\cite{harper2015movielens}, 
Netflix~\cite{bennett2007netflix}, and Amazon 
Beauty~\cite{ni2019justifying}. Statistics are in 
Table~\ref{tab:datasets}. We use an 8:1:1 split for 
train/validation/test following prior 
work~\cite{xia2024hierarchical,xia2024frequency}.
\begin{table}[t]
    \centering
    \caption{Statistics of the four evaluation datasets.}
    \label{tab:datasets}
    \small
    \begin{tabular*}{.8\columnwidth}{@{\extracolsep{\fill}}lrrrr}
        \toprule
        \textbf{Dataset} & \textbf{\#Users} & 
        \textbf{\#Items} & \textbf{\#Inter.} & 
        \textbf{Density} \\
        \midrule
        ML100K & 943 & 1,682 & 100,000 & 6.30\% \\
        ML1M & 6,040 & 3,706 & 1,000,209 & 4.47\% \\
        Netflix & 20,000 & 17,720 & 5,678,654 & 1.60\% \\
        Beauty & 22,363 & 12,101 & 198,502 & 0.07\% \\
        \bottomrule
    \end{tabular*}
    \vspace{-10pt}
\end{table}

\subsubsection{Baselines} 
We compare against \textbf{Popularity} as a baseline without personalization. GCN-based methods: \textbf{LightGCN}~\cite{he2020lightgcn} performs neighborhood aggregation without feature transformation; \textbf{LR-GCCF}~\cite{chen2020revisiting} adds residual connections to graph convolutions; \textbf{IMP-GCN}~\cite{liu2021interest} aggregates within user sub-graphs to reduce noise; \textbf{SimpleX}~\cite{mao2021simplex} optimizes loss functions and negative sampling; \textbf{UltraGCN}~\cite{mao2021ultragcn} approximates infinite-layer graph convolutions. GSP-based methods: \textbf{GF-CF}~\cite{shen2021powerful} applies unified low-pass spectral filtering; \textbf{JGCF}~\cite{guo2023manipulating} amplifies high-frequency signals via Jacobi polynomial filters; \textbf{PGSP}~\cite{liu2023personalized} applies mixed-frequency low-pass filtering over personalized signals; \textbf{HiGSP}~\cite{xia2024hierarchical} designs hierarchical cluster-wise and global filters; \textbf{FaGSP}~\cite{xia2024frequency} combines cascaded and parallel filter modules.
We exclude sequential recommenders because they target next-item prediction, which differs from top-$N$ collaborative filtering in both task formulation and evaluation protocol. \sellname{} uses ordering only at graph construction as shown in Figure~\ref{fig:diffusion-wrap}, and Table~\ref{tab:ablation} isolates this signal's contribution.

\subsubsection{Implementation Details:} 
We fix diffusion depth $d=2$ and decay $\alpha=0.4$ across all datasets, set eigenspace dimensionality $r$ proportional to dataset size, optimize filter parameters $c$ and $w$ via 2D grid search, and tune fusion weight $\phi$ on the validation set. 
Optimal values per dataset are in Table~\ref{tab:filters_summary}. Baselines use official implementations. 
We report results for LR-GCCF, IMP-GCN, SimpleX, UltraGCN, and GF-CF from~\cite{xia2024hierarchical,xia2024frequency}, which use identical preprocessing and evaluation protocols.
\begin{table}[t]
\centering
\caption{Optimal filter parameters per dataset.}
\label{tab:filters_summary}
    \begin{tabular*}{.8\columnwidth}{@{\extracolsep{\fill}}lrrrr}
        \toprule
        \textbf{Dataset} & $\phi$ & $c$ & $w$ & $r$ \\
        \midrule
        ML100K & 0.5 & 0.2 & 0.1 & 32 \\
        ML1M   & 0.3 & 0.8 & 0.1 & 128 \\
        Netflix & 0.3 & 0.4 & 0.3 & 256 \\
        Beauty  & 0.5 & 0.8 & 0.3 & 512 \\
        \bottomrule
    \end{tabular*}
\end{table}

\subsubsection{Metrics:} 
Following~\cite{he2020lightgcn,shen2021powerful, xia2024frequency}, we report NDCG@$k$ and  MRR@$k$ at $k \in \{5,10,20\}$, abbreviated  N@$k$ and M@$k$. 
We exclude previously interacted items from ranking.

\subsection{Overall Performance}

\begin{table*}[t]
\centering
\small
\caption{Performance comparison on four public datasets. The best performance is denoted in \textbf{bold}. ``M@K'' and ``N@K'' are short for ``MRR@K'' and ``NDCG@K'', respectively. ``Improv.'' denotes the percentage improvement of \sellname{} compared to the strongest baseline method.}
\vspace{-5pt}
\label{tab:main_results}
\resizebox{\textwidth}{!}{
\begin{tabular}{llrrrrrrrrrrrrr}
\toprule
\multirow{2.5}{*}{\textbf{Datasets}} & \multirow{2.5}{*}{\textbf{Metric}} & \multicolumn{1}{c}{\textbf{}} & \multicolumn{5}{c}{\textbf{GCN-based Methods}} & \multicolumn{6}{c}{\textbf{GSP-based Methods}} & \multirow{2.5}{*}{\textbf{Improv.}} \\
 \cmidrule(lr){3-3} \cmidrule(lr){4-8} \cmidrule(lr){9-14}
& & Popularity & LightGCN & LR-GCCF & IMP-GCN & SimpleX & UltraGCN & GF-CF & JGCF & PGSP & FaGSP & HiGSP & \sellname{} & \\
\midrule
\midrule
\multirow{6}{*}{ML100K} 
& N@5 & 0.5672 & 0.7034 & 0.5587 & 0.6764 & 0.6995 & 0.6786 & 0.6875 & 0.6582 & 0.6851 & 0.7106 & \underline{0.7166} & \textbf{0.7543} & +5.26\% \\
& M@5 & 0.5151 & 0.6362 & 0.4943 & 0.6006 & 0.6397 & 0.6094 & 0.6234 & 0.4471 & 0.6149 & 0.6518 & \underline{0.6629} & \textbf{0.7358} & +11.00\% \\
& N@10 & 0.5857 & 0.6771 & 0.5603 & 0.6605 & 0.6866 & 0.6688 & 0.6843 & 0.6695 & 0.6722 & \underline{0.7092} & 0.7021 & \textbf{0.7572} & +6.77\% \\
& M@10 & 0.5296 & 0.5877 & 0.4616 & 0.5690 & 0.6064 & 0.5749 & 0.6010 & 0.4222 & 0.5767 & \underline{0.6384} & 0.6380 & \textbf{0.7443} & +16.59\% \\
& N@20 & 0.5867 & 0.6618 & 0.5555 & 0.6605 & 0.6543 & 0.6486 & 0.6621 & 0.6730 & 0.6495 & \underline{0.6989} & 0.6724 & \textbf{0.7465} & +6.81\% \\
& M@20 & 0.5315 & 0.5652 & 0.4411 & 0.5690 & 0.5411 & 0.5289 & 0.5593 & 0.4197 & 0.5496 & \underline{0.5901} & 0.5782 & \textbf{0.7457} & +26.37\% \\
\midrule
\multirow{6}{*}{ML1M} 
& N@5 & 0.0958 & 0.5845 & 0.3540 & 0.5583 & 0.5934 & 0.5739 & 0.5935 & 0.6121 & 0.5963 & \underline{0.6112} & 0.6062 & \textbf{0.6237} & +2.05\% \\
& M@5 & 0.1943 & 0.5160 & 0.2997 & 0.4934 & 0.5244 & 0.5087 & 0.5254 & 0.4172 & 0.5313 & \underline{0.5430} & 0.5386 & \textbf{0.5715} & +5.25\% \\
& N@10 & 0.0960 & 0.5873 & 0.3787 & 0.5594 & 0.5921 & 0.5773 & 0.5897 & \underline{0.6238} & 0.5923 & 0.6082 & 0.6042 & \textbf{0.6431} & +3.10\% \\
& M@10 & 0.2180 & 0.5010 & 0.2946 & 0.4685 & 0.5051 & 0.4886 & 0.4996 & 0.3869 & 0.5063 & \underline{0.5229} & 0.5148 & \textbf{0.5837} & +11.63\% \\
& N@20 & 0.1035 & 0.5692 & 0.3980 & 0.5416 & 0.5687 & 0.5600 & 0.5678 & \underline{0.6270} & 0.5687 & 0.5808 & 0.5769 & \textbf{0.6431} & +2.57\% \\
& M@20 & 0.2309 & 0.4644 & 0.2924 & 0.4292 & 0.4625 & 0.4500 & 0.4557 & 0.3730 & 0.4646 & \underline{0.4748} & 0.4620 & \textbf{0.5871} & +23.65\% \\
\midrule
\multirow{6}{*}{Netflix} 
& N@5 & 0.5105 & 0.6822 & 0.6023 & 0.6328 & 0.6603 & 0.5680 & 0.7005 & 0.7065 & 0.6811 & \underline{0.7162} & 0.7162 & \textbf{0.7201} & +0.54\% \\
& M@5 & 0.4505 & 0.6147 & 0.5323 & 0.5528 & 0.5902 & 0.4941 & 0.6360 & 0.5283 & 0.6127 & 0.6524 & \underline{0.6532} & \textbf{0.6680} & +2.27\% \\
& N@10 & 0.5488 & 0.6814 & 0.6134 & 0.6307 & 0.6694 & 0.5746 & 0.6928 & \underline{0.7148} & 0.6756 & 0.7079 & 0.7062 & \textbf{0.7270} & +1.71\% \\
& M@10 & 0.4712 & 0.5974 & 0.5234 & 0.5386 & 0.5839 & 0.4708 & 0.6126 & 0.4837 & 0.5905 & \underline{0.6293} & 0.6275 & \textbf{0.6762} & +7.45\% \\
& N@20 & 0.5604 & 0.6642 & 0.6043 & 0.6152 & 0.6572 & 0.5580 & 0.6689 & \underline{0.7164} & 0.6543 & 0.6814 & 0.6800 & \textbf{0.7213} & +0.68\% \\
& M@20 & 0.4773 & 0.5636 & 0.4963 & 0.4995 & 0.5555 & 0.4309 & 0.5691 & 0.4273 & 0.5502 & \underline{0.5827} & 0.5800 & \textbf{0.6777} & +16.30\% \\
\midrule
\multirow{6}{*}{Beauty} 
& N@5 & 0.0149 & 0.0668 & 0.0533 & 0.0570 & 0.0682 & 0.0599 & 0.0659 & 0.0501 & 0.0674 & 0.0712 & \underline{0.0718} & \textbf{0.0733} & +2.09\% \\
& M@5 & 0.0118 & 0.0530 & 0.0423 & 0.0444 & 0.0544 & 0.0476 & 0.0523 & 0.0535 & 0.0538 & 0.0564 & \underline{0.0574} & \textbf{0.0619} & +7.84\% \\
& N@10 & 0.0200 & 0.0763 & 0.0635 & 0.0667 & 0.0761 & 0.0649 & 0.0751 & 0.0612 & 0.0765 & \underline{0.0799} & 0.0771 & \textbf{0.0879} & +10.01\% \\
& M@10 & 0.0139 & 0.0519 & 0.0443 & 0.0451 & 0.0517 & 0.0441 & 0.0517 & \underline{0.0590} & 0.0525 & 0.0549 & 0.0531 & \textbf{0.0680} & +15.25\% \\
& N@20 & 0.0256 & 0.0829 & 0.0678 & 0.0729 & 0.0819 & 0.0689 & 0.0819 & 0.0721 & 0.0828 & \underline{0.0846} & 0.0799 & \textbf{0.1013} & +19.74\% \\
& M@20 & 0.0154 & 0.0466 & 0.0388 & 0.0409 & 0.0461 & 0.0393 & 0.0469 & \underline{0.0622} & 0.0481 & 0.0485 & 0.0456 & \textbf{0.0718} & +15.43\% \\
\bottomrule
\end{tabular}
}
\end{table*}
Table~\ref{tab:main_results} shows that \sellname{} consistently outperforms all baselines across all four datasets and all metrics, answering RQ1. Improvements over low-pass GSP methods~\cite{shen2021powerful,liu2023personalized, xia2024frequency,xia2024hierarchical} are consistent with the hypothesis that intermediate-frequency components carry stronger personalization signals than low-frequency popularity trends. 
Improvements over JGCF~\cite{guo2023manipulating}, the only high-frequency baseline, confirm that intermediate-frequency filtering outperforms high-frequency amplification for collaborative filtering.

Performance gains vary inversely with per-user interaction volume. Netflix, whose users have the longest interaction histories, shows the smallest gains (0.54\%--1.71\% on N@$k$). The result suggests that abundant per-user signal is already well-represented in the low-frequency band. In contrast, Beauty, whose users have the shortest histories, shows the largest gains. Limited per-user signal suppresses low-frequency components, so bandpass filtering recovers what low-pass filtering misses. The trend across datasets is consistent with bandpass filtering addressing the failure mode of low-pass methods in sparse settings.

The dataset-specific optimal frequency centers  in Table~\ref{tab:filters_summary} reflect  underlying spectral structure differences.  ML100K concentrates personalization signals at lower  frequencies ($c=0.2$). Beauty and ML1M, with  sparser interactions, require higher bandpass  centers ($c=0.8$) to target the shifted  intermediate-frequency band.
The systematic variation confirms that \sellname{} adapts to  dataset spectral structure rather than imposing  a fixed frequency range.
\subsection{Ablation Analysis}
\begin{table}[t]
    \centering
    \caption{Ablation study on ML1M.}
    \vspace{-10pt}
    \label{tab:ablation}
    \small
    \begin{tabular*}{.8\columnwidth}{@{\extracolsep{\fill}}lrrrr}
        \toprule
        \textbf{Variant} & \textbf{N@10} & $\pm$\textbf{SE} & 
        \textbf{M@10} & $\pm$\textbf{SE} \\
        \midrule
        GSPRec-NB & 0.5769 & 0.0082 & 0.5148 & 0.0093 \\
        GSPRec-NL & 0.6042 & 0.0079 & 0.5229 & 0.0088 \\
        GSPRec-NS & 0.6274 & 0.0039 & 0.5704 & 0.0050 \\
        GSPRec-SE & 0.6416 & 0.0037 & 0.5831 & 0.0049 \\
        \midrule
        \textbf{GSPRec} & \textbf{0.6431} & 0.0074 & 
        \textbf{0.5837} & 0.0081 \\
        \bottomrule
    \end{tabular*}
    \vspace{-10pt}
\end{table}
\begin{figure*}[t]
    \centering
    \begin{subfigure}{\linewidth}
        \centering
        \includegraphics[width=.9\linewidth,keepaspectratio]{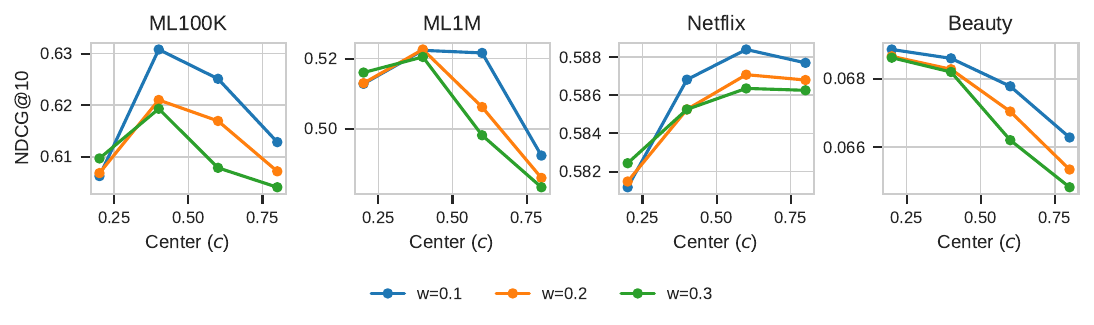}
        \Description{Line plots of NDCG at 10 versus bandpass center c for four datasets, with one line per filter width w. Each curve is flat near its optimum, and the best center differs across datasets.} 
        \vspace{-5pt}
        \caption{Effect of bandpass center position ($c$) on NDCG@10 across datasets. Each line represents a different filter bandwidth ($w$).}
        \label{fig:center_sensitivity}
    \end{subfigure}
    \vspace{5pt}
    \begin{subfigure}{\linewidth}
        \centering
        \includegraphics[width=.9\linewidth,keepaspectratio]{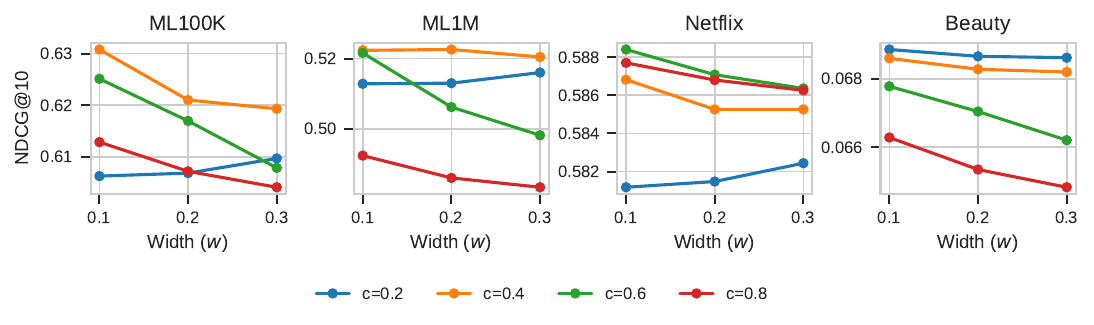}
        \Description{Line plots of NDCG at 10 versus bandpass width w for four datasets, with one line per center position c. Performance varies moderately with width and peaks at dataset-specific settings.}  
        \vspace{-5pt}
        \caption{Effect of bandpass width ($w$) on NDCG@10 across datasets. Each line represents a different center position ($c$).}
        \label{fig:width_sensitivity}
    \end{subfigure}
    \vspace{-10pt}
    \caption{Parameter sensitivity analysis of the Gaussian bandpass filter.}
    \label{fig:sensitivity}
\end{figure*}
Table~\ref{tab:ablation} presents four ablation variants: \textbf{GSPRec-NB} removes the bandpass filter (low-pass only); \textbf{GSPRec-NL} removes the low-pass filter (bandpass only); \textbf{GSPRec-NS} removes item-item proximity from user interaction ordering ($\tilde{\mathbf{S}} = \mathbf{0}$); \textbf{GSPRec-SE} uses direct item-item edges without diffusion ($d=1$, equivalent to $\mathbf{S}'$ only).
Removing the bandpass filter causes the largest performance drop (N@10 falls from $0.6431$ to $0.5769$, $-11.3\%$), answering RQ2. Intermediate-frequency filtering is the primary performance driver. 
GSPRec-NB falls below all GSP baselines in Table~\ref{tab:main_results} because item-item proximity introduces intermediate-frequency structure that the low-pass filter suppresses. Without the complementary bandpass filter, the added structure harms performance.
GSPRec-NS surpasses all GSP baselines, establishing the bandpass filter as a strong independent contribution and answering RQ3. Incorporating item-item proximity adds a further 2.4\% on N@10. The enriched item representations expose structure that the bandpass filter exploits.
The two components are complementary. Neither is sufficient alone and their combination achieves the strongest results. Table~\ref{tab:filters} isolates the filtering contribution. Bandpass alone outperforms low-pass alone by 4.7\% on N@10. The optimal fusion weight $\phi=0.3$ on ML1M indicates that the bandpass component dominates recommendation quality while the low-pass component contributes complementary global structure. 
The depth $d$ keeps diffusion local. Deeper diffusion densifies $\tilde{\mathbf{S}}$ toward uninformative global co-occurrence. The decay $\alpha=0.4$ down-weights two-hop proximity relative to direct closeness. A single setting performs well on all four datasets, which suggests the item-item proximity structure is not dataset-specific.
The improvement from GSPRec-SE to GSPRec confirms that multi-hop diffusion captures meaningful higher-order proximity beyond direct transitions. Eigenspace dimensionality $r$ scales with dataset size, reflecting the richer spectral structure of larger graphs.
\begin{table}[t]
    \centering
    \caption{Filter comparison on ML1M.}
    \vspace{-10pt}
    \label{tab:filters}
    \small
    \begin{tabular*}{.8\columnwidth}{@{\extracolsep{\fill}}lrrrr}
        \toprule
        \textbf{Components} & \textbf{N@10} & 
        $\pm$\textbf{SE} & \textbf{M@10} & 
        $\pm$\textbf{SE} \\
        \midrule
        LP only & 0.5769 & 0.0082 & 0.5148 & 0.0093 \\
        BP only & 0.6042 & 0.0079 & 0.5229 & 0.0088 \\
        LP+BP ($\phi$=0.3) & \textbf{0.6431} & 0.0074 & 
        \textbf{0.5837} & 0.0081 \\
        LP+BP ($\phi$=0.5) & 0.6327 & 0.0076 & 
        0.5714 & 0.0085 \\
        LP+BP ($\phi$=0.7) & 0.6281 & 0.0078 & 
        0.5682 & 0.0087 \\
        \bottomrule
    \end{tabular*}
    \vspace{-10pt}
\end{table}

\subsection{Parameter Sensitivity}
Figures~\ref{fig:center_sensitivity}  and~\ref{fig:width_sensitivity} show sensitivity  curves for $c$ and $w$ respectively. Since $c$  and $w$ interact, varying one while fixing the other does not reproduce the joint optimum. A wider bandwidth at a suboptimal center amplifies noise rather than signal. 
The curves are flat  around the optimum, confirming that \sellname{}  is robust to moderate parameter deviations and  does not require precise hyperparameter tuning.

\subsection{Computational Analysis}
Table~\ref{tab:runtime} shows that \sellname{}  requires 1.27 minutes on ML1M, faster than all  GSP methods except GF-CF (1.18 min) and over  $100\times$ faster than LightGCN (135.96 min).  \sellname{} is the only method that achieves state-of-the-art accuracy at near-GF-CF efficiency.
\sellname{} computes scores directly by spectral filtering rather than learning embeddings for approximate nearest-neighbor top-$K$ retrieval. 
Computing one user's scores requires $O(rn)$ operations on cached eigenvectors, and \sellname{} never stores the full $m \times n$ score matrix.
\begin{table}[t]
    \centering
    \caption{Runtime comparison on ML1M (minutes).}
    \vspace{-10pt}
    \label{tab:runtime}
    \small
    \begin{tabular*}{.7\columnwidth}{@{\extracolsep{\fill}}lr}
        \toprule
        \textbf{Method} & \textbf{Runtime (min)} \\
        \midrule
        LightGCN~\cite{he2020lightgcn} & 135.96 \\
        SimpleX~\cite{mao2021simplex} & 109.83 \\
        UltraGCN~\cite{mao2021ultragcn} & 16.09 \\
        GF-CF~\cite{shen2021powerful} & 1.18 \\
        PGSP~\cite{liu2023personalized} & 2.36 \\
        FaGSP~\cite{xia2024frequency} & 2.08 \\
        HiGSP~\cite{xia2024hierarchical} & 2.04 \\
        \midrule
        \textbf{GSPRec (Ours)} & \textbf{1.27} \\
        \bottomrule
    \end{tabular*}
\end{table}
\section{Conclusion}
Prior GSP-based CF methods derive item representations from the interaction matrix alone and filter them with low-pass designs, discarding the intermediate-frequency band where community-level preferences reside.
\sellname{} incorporates ordering-derived item-item proximity into the graph topology.
The enriched graph exposes intermediate-frequency structure that a Gaussian bandpass filter selectively amplifies.
A complementary low-pass filter retains popularity trends.
\sellname{} outperforms all graph CF baselines across four datasets.
Ablations show that graph construction and filter design are coupled, since item-item proximity without the bandpass filter degrades performance.
Future directions include adapting the spectral band under distribution shift and approximating the filters with polynomials for industrial-scale deployment.


\bibliographystyle{ACM-Reference-Format}
\bibliography{references}

\end{document}